\documentclass[preprint]{sig-alternate-05-2015}

\usepackage{amsmath}
\usepackage[compatibility=false,font=scriptsize,labelfont=bf]{caption}
\usepackage{epsfig,epsf,rotating}
\usepackage{graphicx}
\usepackage{amssymb}
\usepackage{epstopdf}
\usepackage{amsmath}
\usepackage{subscript}
\usepackage{amsfonts}
\usepackage[latin1]{inputenc}
\usepackage{amssymb}
\usepackage{latexsym}
\usepackage{float}
\usepackage[tight,TABTOPCAP, scriptsize]{subfigure} 
\usepackage{rotating}
\usepackage{verbatim}
\usepackage{multicol}
\usepackage{hyperref}
\usepackage{rotating}
\usepackage{afterpage}
\usepackage{balance}
\usepackage{listings}
\usepackage{color}
\usepackage{multirow}
\usepackage{textcomp}
\usepackage{enumitem}
\usepackage{mathtools}
\usepackage{setspace}
\usepackage{cancel}
\usepackage[ruled,noend,linesnumbered,commentsnumbered]{algorithm2e}
\usepackage{epsf}
\usepackage{enumitem}
\DeclareCaptionType{copyrightbox}

\usepackage{theorem}

\theoremstyle{definition}  
 
\newtheorem{lemma}{Lemma}[section]
\newtheorem{theorem}{Theorem}[section]

\newtheorem{fact}{Fact}[section]

\newcommand{\hmcs}[1]{HMCS$\langle #1\rangle$}
\newcommand{\qnode}{\texttt{QNode}}
\newcommand{\qnodes}{\texttt{QNode}s}
\newcommand{\thr}{$t$}

\newcommand{\cas}{\texttt{CAS}}
\newcommand{\swap}{\texttt{SWAP}}
\newcommand{\NULL}{\texttt{null}}
\newcommand{\succe}{$s$}
\newcommand{\pred}{$p$}

\newcommand{\recycled}[1][]{\texttt{R}$_{#1}$}
\newcommand{\wait}[1][]{\texttt{W}$_{#1}$}
\newcommand{\abandoned}[1][]{\texttt{A}$_{#1}$}
\newcommand{\unlocked}[1][]{\texttt{U}$_{#1}$}
\newcommand{\cohort}[1][]{\texttt{C}$_{#1}$}
\newcommand{\parent}[1][]{\texttt{P}$_{#1}$}
\newcommand{\edge}[2]{edge $#1\rightarrow#2$}
\newcommand{\boxedHMCST}{\mbox{HMCS-T}}

\newcommand{\hmcst}[1]{\boxedHMCST{}$\langle #1\rangle$}

\makeatletter
\def\@copyrightspace{\relax}
\makeatother
\begin{document}

\setcopyright{acmcopyright}

\doi{10.475/123_4}

\isbn{123-4567-24-567/08/06}

\conferenceinfo{PLDI '13}{June 16--19, 2013, Seattle, WA, USA}

\acmPrice{\$15.00}

%
\conferenceinfo{WOODSTOCK}{'97 El Paso, Texas USA}

\title{Correctness of Hierarchical MCS Locks with Timeout}
%
%
%
%
%
\numberofauthors{3} 
%
\author{
%
%
\alignauthor
Milind Chabbi\\
       \affaddr{Hewlett Packard Labs}\\
      \affaddr{Palo Alto, CA}\\
       \email{milind.chabbi@hpe.com}   
\alignauthor
Abdelhalim Amer\\
       \affaddr{Argonne National Laboratory}\\
      \affaddr{Lemont, IL}\\
       \email{aamer@anl.gov}   
\alignauthor
Shasha Wen, Xu Liu\\
      \affaddr{College of William and Mary}\\
      \affaddr{Williamsburg, VA}\\
      \email{\{swen, xl10\}@cs.wm.edu}
}


\maketitle

%
%


%
%

%
%



\begin{figure*}
\centering
\begin{minipage}{\linewidth}
\includegraphics[width=\linewidth]{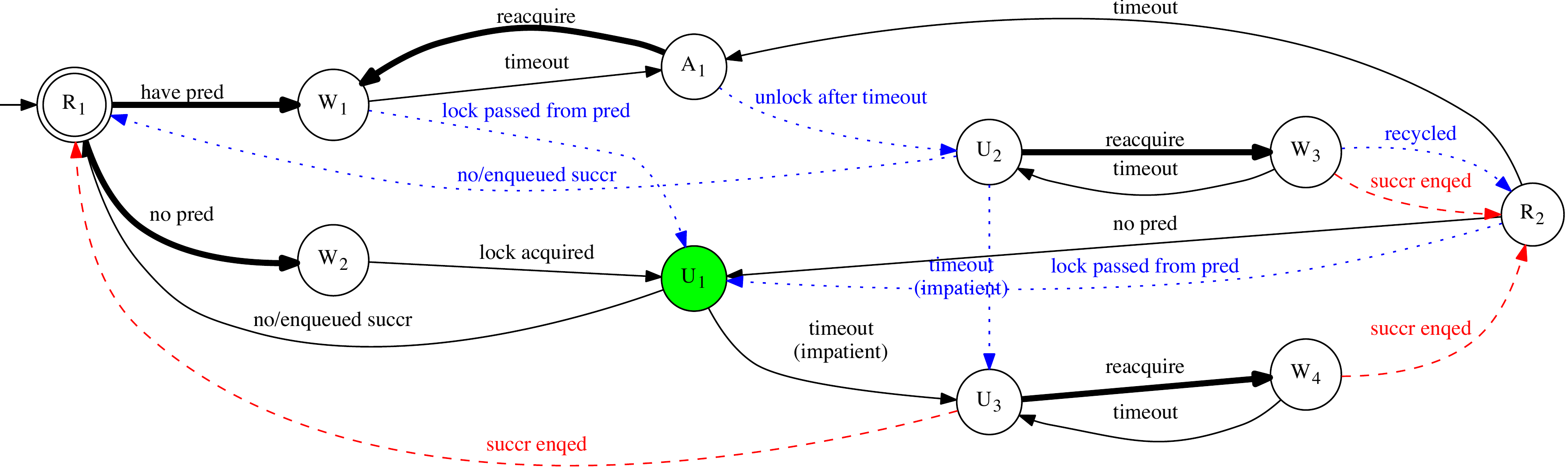}
\caption{{NFA for a \qnode{} status field in \hmcst{1}.}}
\label{fig:NFAStatusOneLevel}
\end{minipage}
\begin{minipage}{\linewidth}
\centering
\includegraphics[width=\linewidth]{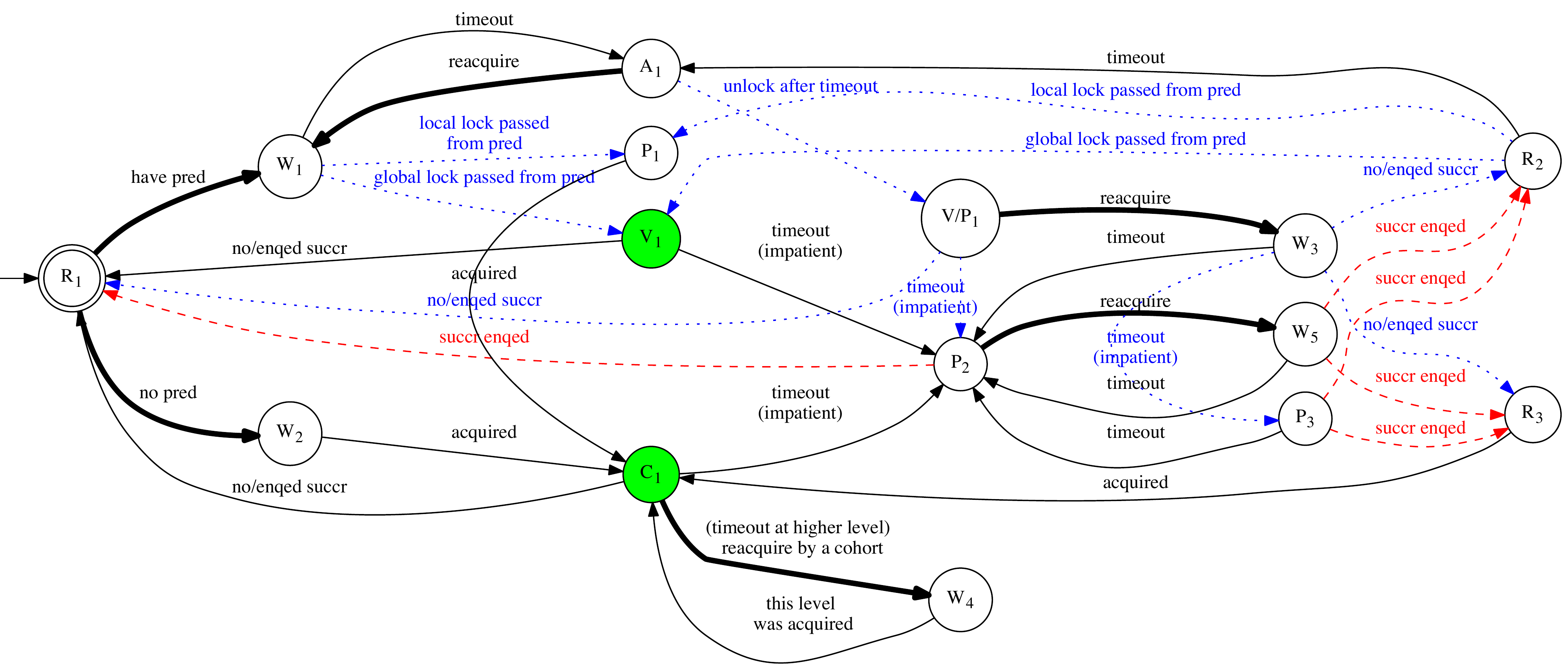}
\caption{{NFA for the status field of a non-root-level  \qnode{}.}}
\label{fig:NFAStatusNonRootLevel}
\end{minipage}
\begin{minipage}{\linewidth}
\centering
\begin{minipage}{.8\linewidth}
\includegraphics[width=\linewidth]{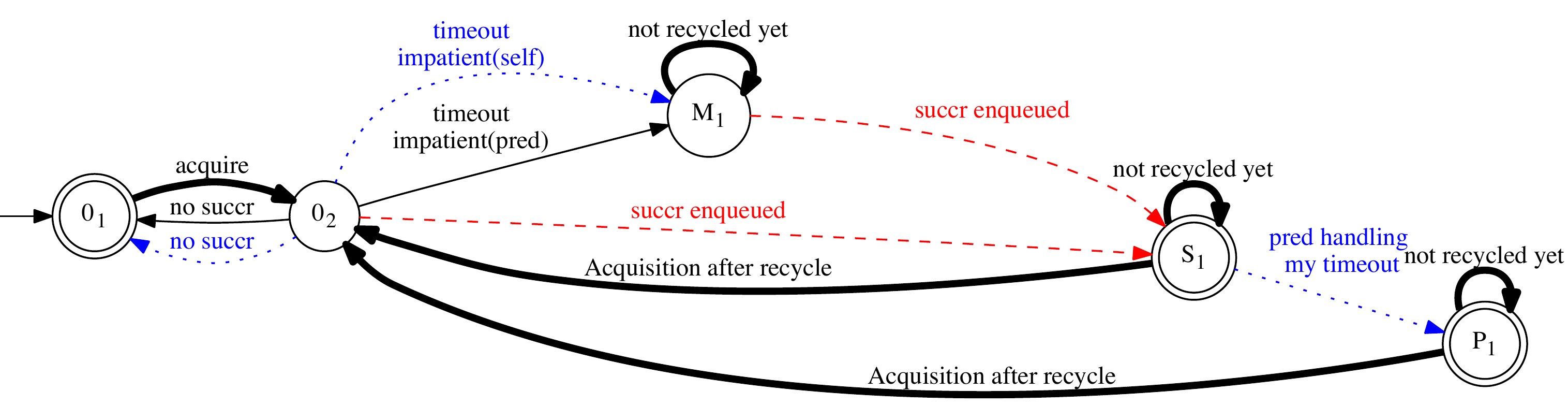}
\caption{{NFA for a \qnode{} \texttt{next} field in \hmcst{n}. There is no designated ''lock acquired'' node.}}
\label{fig:NFANextOneLevel}
\end{minipage}
\centering
\includegraphics[width=\linewidth]{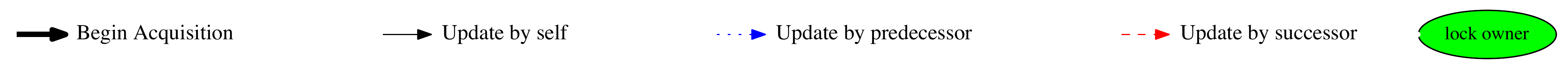}
\caption{{Legend for Figures~\ref{fig:NFAStatusOneLevel},~\ref{fig:NFAStatusNonRootLevel}, and~\ref{fig:NFANextOneLevel}}}
\end{minipage}
\end{figure*}

This manuscript serves as a correctness proof of the Hierarchical MCS locks with Timeout (HMCS-T) described in our paper~\cite{Chabbi:2017:Abort} titled ``An Efficient Abortable-locking Protocol for Multi-level NUMA Systems,'' appearing in the proceedings of the 22nd ACM SIGPLAN Symposium on Principles and Practice of Parallel Programming.

\boxedHMCST{} is a very involved protocol. The system is stateful; the values of prior acquisition efforts affect the subsequent acquisition efforts.
Also, the status of successors, predecessors, ancestors, and descendants affect steps followed by the protocol.
The ability to make the protocol fully non-blocking leads to modifications to the \texttt{next} field, which causes deviation from the original MCS lock protocol both in acquisition and release.
At several places, unconditional field updates are replaced with \swap{} or \cas{} operations.

We follow a multi-step approach to prove the correctness of \boxedHMCST{}. 
To demonstrate the correctness of \boxedHMCST{} lock, we make use of the Spin~\cite{Holzmann:1997:MCS:260897.260902} model checking. 
Model checking causes a combinatorial explosion even to simulate a handful of threads. 
First we understand the minimal, sufficient configurations necessary to prove safety properties of a single level of lock in the tree.
We construct \boxedHMCST{} locks that represent these configurations.
We model check these configurations, which proves the correctness of components of an \boxedHMCST{} lock.
Finally, building upon these facts, we argue logically for the correctness of \hmcst{n}.

\section{Minimal configuration}
We need to answer the following questions to design an \boxedHMCST{} lock configuration that is sufficient to exercise all possible thread interleaving in any arrangement: 

\begin{itemize}
\item How many threads are sufficient?
\item How many lock levels are sufficient?
\item How many lock acquisitions per participant are sufficient?
\end{itemize}

To answer these questions, we build  non-deterministic finite acceptors (NFAs) that capture the state transition for each shared variable.
The shared variables are the \texttt{status} and \texttt{next} fields of a \qnode{} and the \texttt{tail pointer} variable.
The transitions of the status flag of a  root-level \qnode{} are different from the transitions of the status field of a non-root-level \qnode{}.
Figure~\ref{fig:NFAStatusOneLevel} and Figure~\ref{fig:NFAStatusNonRootLevel}, respectively, show the  NFA for the \texttt{status} field of a root-level and a non-root-level \qnode{}.
Figure~\ref{fig:NFANextOneLevel} shows the NFA for the \texttt{next} field of any \qnode{}.
The \texttt{tail pointer}  variable can be either \NULL{} or non-\NULL{}, and it is less interesting in designing the \boxedHMCST{} verification configurations.
Appendix~\ref{appx:NFARoot},~\ref{appx:NFANonRoot}, and~\ref{appx:NFANext} describe the transition associated with every edge shown in Figure~\ref{fig:NFAStatusOneLevel},~\ref{fig:NFAStatusNonRootLevel}, and~\ref{fig:NFANextOneLevel}, respectively.

Node labels in Figure~\ref{fig:NFAStatusOneLevel}-\ref{fig:NFAStatusNonRootLevel} represent the field values in those states, and the subscripts distinguish the same values that bear different meanings in different contexts.
Solid black edges represent the actions taken by a thread \thr{} owning the \qnode{} under scrutiny.
Dotted blue edges represent the actions taken by a predecessor \pred{} of \thr{}.
Dotted red edges represent the actions taken by a successor \succe{} of \thr{}.
Thick black edges represent beginning of a new acquisition effort by a thread \thr{} that owns the \qnode{}.
Any subsequent path formed only of solid black edges represents a sequence of actions taken by a same thread of execution. 
Since the first operation in any acquisition is \swap{}ing the status field, every new acquisition edge has a \wait[i] node as its sink.
Green color filled node(s) represent the state(s) where the lock contending thread \thr{} has become the  owner of the lock at that level.

The NFA provides the following key insights:
\begin{enumerate}
\item \textbf{Three participants:} 
Any edge can be traversed via a path starting at the start state that involves no more than a predecessor (dotted blue edge), self (black edge), and a successor (dotted red edge) in Figures~\ref{fig:NFAStatusOneLevel}, ~\ref{fig:NFAStatusNonRootLevel}, and~\ref{fig:NFANextOneLevel}.
Hence, three participants (a predecessor, self, and a successor) are sufficient to exercise all possible transitions that the status field of a \qnode{} may go through. 

\item \textbf{Two rounds:}  
Any edge can be traversed via a path starting at the start state that involves no more than two ``begin acquisition'' (thick black line) edges.
Hence, two rounds of acquisitions on the same \qnode{} are sufficient to exercise all possible transitions. This means, at least, one thread should try two acquisitions.
The other two threads can perform one acquisition each to exercise all interleaving of the third thread that performs two acquisitions.

\item \textbf{Three levels:}  
The \edge{C_1}{W_4}  in Figure~\ref{fig:NFAStatusNonRootLevel} demands that a thread $t_1$ to have acquired the lock at the current node \texttt{q} at level $l$ and abandoned at an ancestor level and a different thread  $t_2$, a peer of $t_1$ at a level $<l$, to have inherited the level $l$ lock from $t_1$.
Hence, there should,  at least, be two threads at level $l-1$, which can cause one of them (say $t_1$) to acquire locks at level $l-1$ and $l$ but timeout at level $l+1$ and eventually grant the locks at level $l-1$ and $l$ to another thread (say $t_2$). 
Three levels, parent, current, and children are sufficient to exercise all possible transitions in a non-root-level \qnode{} .
\end{enumerate}

To elaborate on Property 1 and 2, we describe a few interesting transitions in  Figure~\ref{fig:NFAStatusOneLevel}.
The \edge{U_2}{U_3} needs \thr{} to have a predecessor to reach \unlocked[2] and then a successor to cause impatience  during the release protocol to transition to \unlocked[3].
The \edge{W_3}{R_2} needs \thr{} to have a predecessor to reach \unlocked[2] and then the second round of acquisition attempt by \thr{} to reach \wait[3] and then a successor to make \thr{} impatient in its release protocol to eventually make the successor update \thr{}'s status to \recycled[2].  
The \edge{R_2}{U_1} and \edge{R_2}{A_1}  need at least two rounds of acquisitions by \thr{} and a successor \succe{} to reach \recycled[2]. The same successor \succe{} can act as a predecessor for \edge{R_2}{U_1} transition. Similarly, \succe{} can act as a predecessor leading to a timeout to cause \edge{R_2}{A_1} transition.

\begin{figure}[t]
\centering
\begin{minipage}{.38\linewidth}
\includegraphics[width=\linewidth]{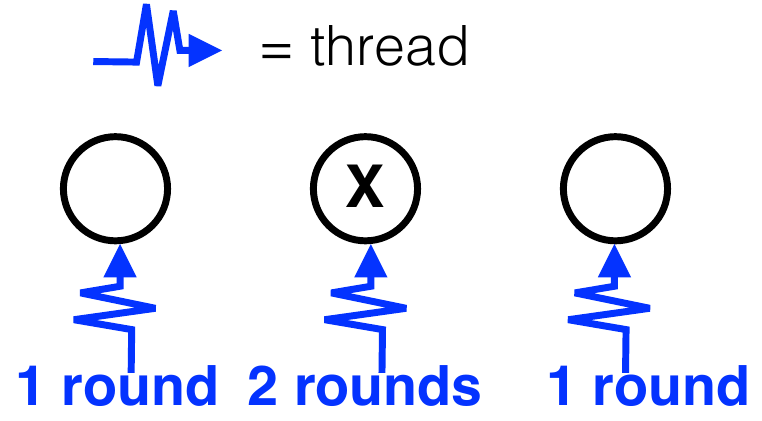}
\caption{{Model checking configuration to exercise all possible interleaving for a thread at root level.}}
\label{fig:SpinRoot}
\end{minipage}
\hspace{1ex}
\centering
\begin{minipage}{.58\linewidth}
\includegraphics[width=\linewidth]{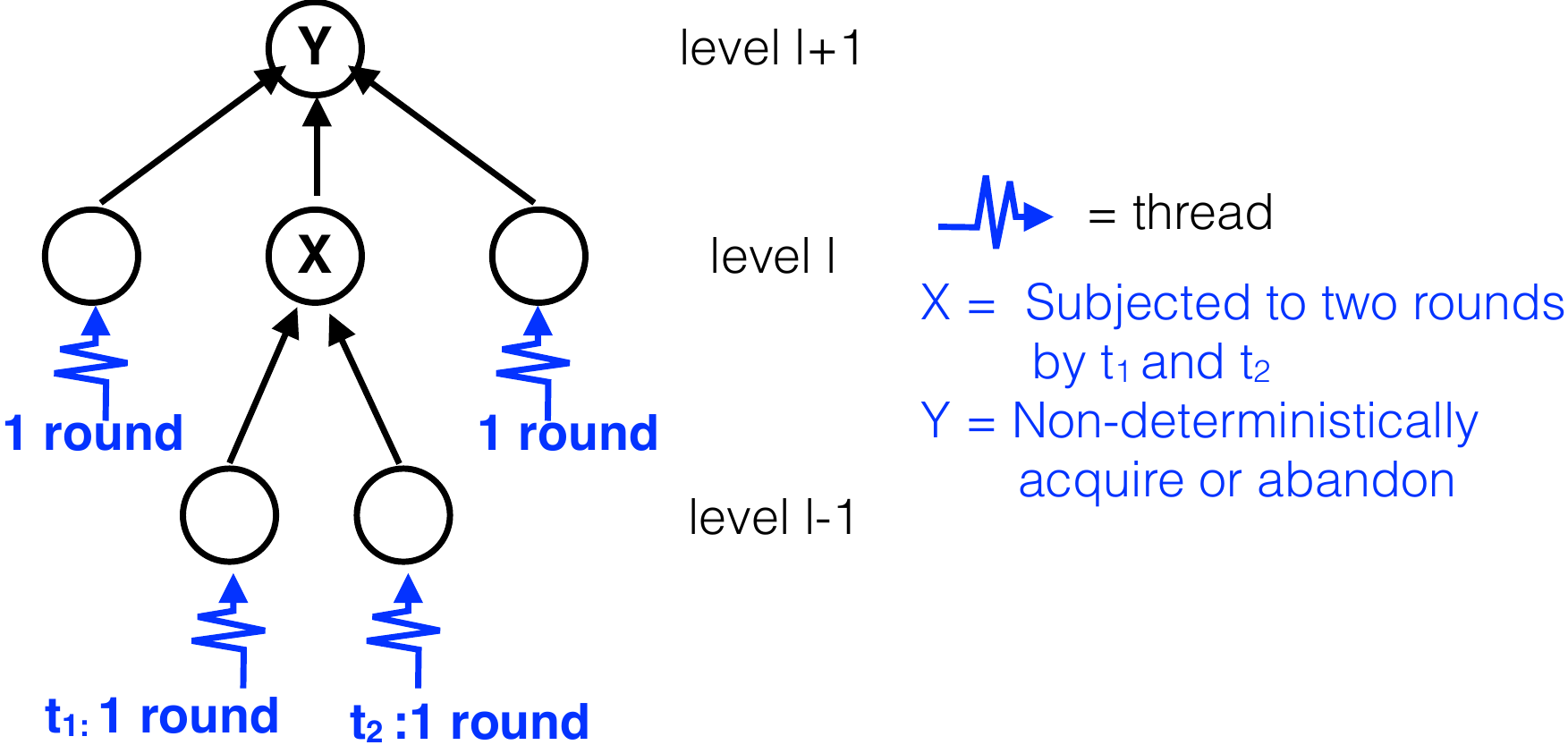}
\caption{{Model checking configuration to exercise all possible interleaving for a thread at a non-root level.}}
\label{fig:SpinNonRoot}
\end{minipage}
\end{figure}

NFAs, unfortunately, do not capture an important safety property---mutual exclusion.
An NFA is ill-defined if the ownership of a \qnode{} is not exclusive, which can happen if another thread belonging to the same domain starts modifying a shared \qnode{}. 
To check the mutual exclusion property, we exercise all possible thread interleaving in a model checking phase.

To exercise all states of the root-level lock we use a thread configuration shown in Figure~\ref{fig:SpinRoot}. 
The thread under scrutiny will be subjected to two rounds of acquisitions and the other two threads perform one round of acquisition each.
Since model checking will exercise all interleaving, the timeout value is immaterial.

To exercise all states of the non-root-level lock, we use a thread configuration shown in Figure~\ref{fig:SpinNonRoot}. 
There are two threads at level $1$, which can causes one of them (say $t_1$) to acquire the locks at level $1$ and $2$ but timeout at level $3$ and eventually grant the ownership of locks at level $1$ and $2$ to another thread (say $t_2$).
The presence of two threads at level $1$, also causes the common ancestor $X$, the \qnode{} under scrutiny at level $l$, to go through the necessary two rounds of acquisitions. 
The other two participants---a successor \succe{}, and a predecessor \pred{} at level $l$--- perform only one round of acquisition each.
The model checking does not require \succe{} and  \pred{}  to begin the protocol at the leaf level, which avoids exercising some non-interesting interleavings.
Hence, we set up \succe{} and \pred{} without children.
Note that such arrangement is for model checking only; the \boxedHMCST{} lock admits new acquisitions starting at the leaf level only.
In total, we need 4 threads, 2 at level $1$ sharing the parent  $X$, and 3 (of which one would have ascended from $1$) at level $2$.
The behavior at level $3$ will be non-deterministic---either a successful acquisition or  abandonment to simulate all possible transitions in $X$.
Non-deterministic behavior is easy to exhibit in Spin~\cite{Holzmann:1997:MCS:260897.260902}. 

The verification checks for the assertion that two threads are never simultaneously in the critical section for the configuration in Figure~\ref{fig:SpinRoot}.
This assertion ensures that the root-level lock ensures mutual exclusion to the critical section if each \qnode{} is accessed by  descendent threads in a mutually exclusive manner.
For the configuration in Figure~\ref{fig:SpinNonRoot}, we check that $t_1$ and $t_2$ never simultaneously acquire the level $l-1$ lock and no two threads ever simultaneously acquire the level $l$ lock.
This assertion ensures that a non-root-level lock ensures mutual exclusion to its next level if each \qnode{} is accessed by descendent threads in a mutually exclusive manner.

Additionally, the NFAs in Figure~\ref{fig:NFAStatusOneLevel},~\ref{fig:NFAStatusNonRootLevel}, and~\ref{fig:NFANextOneLevel} provide insights into the following key properties:
\begin{enumerate}
\item \textbf{Livelock Freedom:} There does not exist any cycle without at least one ``begin new acquisition'' edge. Hence, there cannot be perpetual state transitions (live lock) without user opting to start another round of lock acquisition.
\item \textbf{Starvation Freedom:} Every \wait[i] node (beginning of a new acquisition) has a path to the lock owning state (\unlocked[1] in Figure~\ref{fig:NFAStatusOneLevel} and \texttt{V}$_1$ and \cohort[1] in Figure~\ref{fig:NFAStatusNonRootLevel}), if it is not allowed to traverse any \emph{timeout} edge. This implies, every thread that starts its acquisition process and does not timeout, eventually acquires the lock.
The \texttt{next} field does not decide the lock ownership and hence ignored.
\item \textbf{Bounded Steps to Release:} There exists a finite-length solid-black edge path from lock owner state to another node $\eta$ such that a new acquisition (thick black edge) effort can begin at $\eta$.
 This implies, 1) an acquired lock can be released in a bounded number of steps by the lock owner and 2) once the lock is released, the \qnode{} can be subjected to another acquisition attempt immediately. 
\item \textbf{Bounded Steps on Timeout:} Every node that is \emph{not}  source node of a new acquisition edge (thick black edge) has a solid-black edge path to the source of a timeout edge.
This implies that in any state after starting an acquisition process if a timeout occurs, \thr{} can abandon the protocol in a bounded number of steps.
Source nodes of new acquisition edges are precluded because one cannot start an abandonment without having started an acquisition.
\item \textbf{Deadlock Freedom:} Every node has a path (there is an $\epsilon$ path to itself) formed out of solid-black edges to a node from where a new acquisition can begin.
\end{enumerate}

\section{Correctness of \hmcst{n}}

To establish the mutual exclusion guarantee of \hmcst{n}, we take the following steps:

\begin{lemma}\emph{\textbf{(Root level lock ensures mutual exclusion:)}}  A root-level lock ensures mutual exclusion if every root-level \qnode{} is owned by a descendent in a mutually exclusive manner.
\label{lem:root}
\end{lemma}

\begin{proof}
Verified by model checking a root-level lock with the configuration shown in Figure~\ref{fig:SpinRoot}.
\end{proof}

\begin{lemma} \emph{\textbf{(Non-root level lock ensures mutual exclusion:)}} A non-root-level lock admits mutually exclusive access to the next level lock if every \qnode{} at that level is owned by a single descendent at a time.
\label{lem:nonroot}
\end{lemma}
\begin{proof}
Verified by model checking a non-root-level in an \boxedHMCST{} lock with the configuration shown in Figure~\ref{fig:SpinNonRoot}.
\end{proof}

\begin{fact}[Exclusive ownership of leaf-level node:] Every \qnode{} at leaf level is owned by a unique thread, and the ownership is never shared with any other thread.
\label{axiom:exclusiveleaf}
\end{fact}

\begin{theorem} [\boxedHMCST{} ensures mutual exclusion:] \hmcst{n} ensures mutual exclusion to the critical section it protects.
\label{thm:hmcsnExl}
\end{theorem}
\begin{proof}
  \hmcst{n} is composed of a root-level lock and $n-1$ non-root-level locks.  
Each level ensures mutual exclusion to the level above as long the threads from descendent levels (if any) accesses the shared \qnode{} at the current level in a mutually exclusive manner. 
Assume \hmcst{n} does not ensure mutual exclusion to the critical section. This means two threads $t_1$ and $t_2$ can  simultaneously be in the critical section.
Both $t_1$ and $t_2$ are either 1) peers at level $n$ and hence compete for the root-level lock at level $n$, or 2)  belong to the same domain and hence compete for a non-root-level lock at  a level $l < n$.

If $t_1$ and $t_2$  are peers at level $n$, they will enqueue, two different \qnodes{} and compete for the root-level lock and by Lemma~\ref{lem:root} only one of them can be in the critical section at a time.
Hence, $t_1$ and $t_2$ cannot be peers at the root-level.

Now, $t_1$ and $t_2$ are either peers at level $n-1$ or belong to the same domain at level $l' < n-1$.
If $t_1$ and $t_2$  are peers at level $n-1$, they will enqueue two different \qnodes{} and compete for the non-root-level lock at level $n-1$ and by Lemma~\ref{lem:nonroot} only one of them can own the level $n-1$ lock ensuring the mutual exclusion between them.
Hence, $t_1$ and $t_2$ cannot be peers at level $n-1$.

Since \hmcst{n} has only a finite number of levels, by extrapolation, $t_1$ and $t_2$ are either peers at the leaf level or  share the same \qnode{} at the leaf level.
If $t_1$ and $t_2$  are peers at the leaf level, they will enqueue two different \qnodes{} and compete for the non-root-level lock at the leaf level and by Fact~\ref{lem:nonroot} only one of them can own the leaf level lock ensuring the mutual exclusion between them.
Hence, $t_1$ and $t_2$ must be sharing the same \qnode{} at the leaf level. 
By Lemma~\ref{axiom:exclusiveleaf}, no two threads can share the same \qnode{} at the leaf level, hence $t_1 = t_2$, which contradicts the assumption.

Hence, only one thread can be in the critical section in \hmcst{n}.
\end{proof}

The desirable attributes---starvation freedom, live-lock and deadlock freedom, bounded steps to release or time out---for a given level of lock do not translate to the same for an entire \hmcst{n} lock.
To establish these properties for \hmcs{n}, we make the following claims:
\begin{fact}[Ordered acquisition:] Any thread in \boxedHMCST{} lock of $n$ levels obeys a monotonically increasing order in acquisition effort starting from level $1$ and ending at level $l \le n$.
\label{axiom:acquisition}
\end{fact}
\sloppy
\begin{fact}[Ordered release and abandonment:] \boxedHMCST{} lock of $n$ levels obeys a bitonically ordered release and abandonment---monotonically increasing in level followed by monotonically decreasing in level. 
A thread owning  locks  $1<=$ \texttt{prefix:suffix} $ \le n$ either releases the \texttt{suffix} locks before releasing the ownership of remaining \texttt{prefix} locks or delegates the same responsibility to  another thread that becomes the owner of entire \texttt{prefix:suffix} locks. 
\label{axiom:release}
\end{fact}

\begin{theorem} \hmcst{n} guarantees live-lock freedom, deadlock freedom, starvation freedom, bounded steps to release, and bounded steps on timeout.
\end{theorem}

\begin{proof}
\hmcst{n} is composed of a root-level lock and $n-1$ non-root-level locks.  
By Fact~\ref{axiom:acquisition} and~\ref{axiom:release}, every thread follows an ordered acquisition and release  or abandonment protocol.
Hence, each thread goes through a finite number of levels in any process.
At each level, root or non-root, the NFA that a thread is subjected to for its \qnode{}, ensures live-lock freedom, deadlock freedom, starvation freedom, bounded steps to release, and bounded steps on timeout if the \qnode{} is accessed mutually exclusively by descendants that share the same ancestor \qnode{}.
By Theorem~\ref{thm:hmcsnExl}, each \qnode{} is owned by a descendent thread in a mutually exclusive manner. 
Hence, by construction \hmcst{n} ensures live-lock freedom, deadlock freedom, starvation freedom, bounded steps to release, and bounded steps on timeout. 
\end{proof}

\appendix
\section{NFA for the status field of a root-level QNode}
\label{appx:NFARoot}
The status always starts in \recycled[1] state. 
\textbf{All other states are transient; a correctly implemented \hmcst{1} ought to revert the status of very \qnode{} to \recycled[1] eventually.}
On a fresh acquisition in the \recycled[1] state of a \qnode{} \texttt{q}, the initial \swap{} on \texttt{q.status} moves it non-deterministically to either \wait[1] (if there was a predecessor) or \wait[2] (no predecessor).

If no predecessor, the thread \thr{} updates \texttt{q.status} to \unlocked[1] (\edge{W_2}{U_1}).
In \unlocked[1], if \thr{} has a successor \succe{} that has already advertised itself with \texttt{q.next} or there is no successor, \thr{} releases the lock and updates \texttt{q.status} to \recycled[1] (\edge{U_1}{R_1}).
In \unlocked[2], if  \thr{} leaves due to timeout because a successor \succe{}  has not updated \texttt{q.next}, the NFA transitions into state \unlocked[3] (\edge{U_1}{U_3}).
In \unlocked[3], if \succe{} advertises itself and recycles \texttt{q.status}, the NFA transitions to \recycled[1] (\edge{U_3}{R_1}).
In \unlocked[3], if \thr{} attempts to re-acquire the lock, it will \swap{} \texttt{q.status} to \wait[4] (\edge{U_3}{W_4}).
If \thr{} times out in \wait[4] while waiting for it  to become \recycled{}, it reverts the state back to \unlocked[3]  (\edge{W_4}{U_3}).
In \wait[4], if \succe{} advertises itself and recycles \texttt{q.status}, the NFA transitions to \recycled[2] (\edge{W_4}{R_2}).

In \wait[1], a predecessor may pass the lock to the waiting thread \thr{} updating \texttt{q.status} to \unlocked[1] (\edge{W_1}{U_1}).
If \thr{} times out in \wait[1], it updates the state to \abandoned[1]  (\edge{W_1}{A_1}).
In  \abandoned[1], a predecessor \pred{} may move the status to \unlocked[2] (\edge{A_1}{U_2}).
In  \abandoned[1], any attempt by \thr{} to re-acquire the lock reverts the state to \wait[1] (\edge{A_1}{W_1}).
In \unlocked[2], if \pred{} manages to successfully release the lock, it will eventually transition \texttt{q.status} to \recycled[1] (\edge{U_2}{R_1}).
In \unlocked[2], if \pred{} times out (impatient) waiting for a successor delayed in updating  \texttt{q.next} field, the NFA transitions to \unlocked[3] (\edge{U_2}{U_3}).
In \unlocked[2], any attempt by \thr{}  to re-acquire the lock moves the state to \wait[3] (\edge{U_2}{W_3}).
If \thr{} times out in \wait[3], it reverts the state to \unlocked[2]  (\edge{W_3}{U_2}).
In \wait[3], either a predecessor may update the state to recycled \recycled[2], or an impatient predecessor may time out and a successor may update the state to recycled \recycled[2] (\edge{W_3}{R_2}).

In \recycled[2], \thr{} will reenqueue the \qnode{} and it may acquire the lock via transition to \unlocked[1] either because it has no predecessors or a predecessor passed the lock (\edge{R_2}{U_1}).
In \recycled[2], after enqueuing the node, if \thr{} times out waiting for the lock, it will transition to \abandoned[1] (\edge{R_2}{A_1}).
\section{NFA for the status field of a non-root-level QNode}
\label{appx:NFANonRoot}
We now describe the state diagram for the \texttt{status} field of a non-root-level \qnode{}.

The status always starts in \recycled[1] state. 
\textbf{All other states are transient, a correctly implemented non-root-level ought to revert the status of very \qnode{} to \recycled[1] eventually.}
On a fresh acquisition in the \recycled[1] state of a \qnode{} \texttt{q}, the initial \swap{} on \texttt{q.status} moves it non-deterministically to either \wait[1] (if there was a predecessor) or \wait[2] (no predecessor).

If no predecessor, the thread \thr{} updates \texttt{q.status} to \cohort[1] (\edge{W_2}{C_1}).
IN \cohort[1], if \thr{} has a successor \succe{} that has already advertised itself with \texttt{q.next} or there is no successor, \thr{} releases the lock and updates \texttt{q.status} to \recycled[1] (\edge{C_1}{R_1}).
In \cohort[1], if  \thr{} leaves due to timeout because a successor \succe{}  has not updated \texttt{q.next}, \thr{} leaves \texttt{q} by updating its status to \parent[2]  (\edge{C_1}{P_2}).
In \parent[2], if \succe{} advertises itself and recycles \texttt{q.status}, the NFA transitions to \recycled[1] (\edge{P_2}{R_1}).
In \parent[2], if \thr{} attempts to re-acquire the lock, it will \swap{} \texttt{q.status} to \wait[5] (\edge{P_2}{W_5}).
If \thr{} times out in \wait[5] while waiting for it  to become \recycled{}, it reverts the state back to \parent[2]  (\edge{W_5}{P_2}).
In \wait[5], if \succe{} advertises itself and recycles \texttt{q.status}, the NFA non-deterministically transitions to either \recycled[3] (\edge{W_5}{R_3}, if it finds no predecessor by the time \thr{} re-enqueues the node) or to \recycled[2] (\edge{W_5}{R_2}, if a predecessor is present by the time \thr{} re-enqueued the node).
In \recycled[3], \thr{} will acquire the lock immediately and update the status to \cohort[1] (\edge{R_3}{C_1}).

In \cohort[1], having acquired the current level (say $l$) lock \thr{} may ascend to an ancestor level and it may abandon the lock at that level.
In an effort to release the locks already held, \thr{} may pass its locks including $l$ lock to another thread, say \thr{}$_2$.
When \thr{}$_2$ begins its acquisition process at level $l$, it will \swap{} \texttt{q.status} to \wait[4] (\edge{C_1}{W_4}) and immediately realize that it inherited this lock and revert \texttt{q.status}  to \cohort[1] (\edge{W_4}{C_1})

If \thr{} times out in \wait[1], it updates the state to \abandoned[1]  (\edge{W_1}{A_1}).
In  \abandoned[1], a predecessor \pred{} may attempt to pass all locks it holds (\texttt{V}, a legal lock passing value) or only a prefix of locks (\parent{}) (\edge{A_1}{V/P_1}).
In  \abandoned[1], any attempt by \thr{} to re-acquire the lock reverts the state to \wait[1] (\edge{A_1}{W_1}).
In \texttt{V/P}$_1$, if \pred{} manages to successfully release the lock, it will eventually transition \texttt{q.status} to \recycled[1] (\edge{V/P_1}{R_1}).
In \texttt{V/P}$_1$, if \pred{} times out (impatient) waiting for a successor delayed in updating  \texttt{q.next} field, the NFA transitions to \parent[2] (\edge{V/P_1}{P_2}).

In \wait[1], a predecessor may pass the global lock (all locks on path to the root)  to \thr{} by  updating \texttt{q.status} to a legal passing value \texttt{V}$_1$  (\edge{W_1}{V_1}).
In \texttt{V}$_1$,  if \thr{} has a successor \succe{} that has already advertised itself with \texttt{q.next} or there is no successor, \thr{} releases the lock and updates \texttt{q.status} to \recycled[1] (\edge{V_1}{R_1}).
In \texttt{V}$_1$, if  \thr{} leaves due to timeout because a successor \succe{}  has not updated \texttt{q.next}, \thr{} would have already released all ancestral locks and then it leaves \texttt{q} by updating \texttt{q.status } to \parent[2]  (\edge{V_1}{P_2}).
In \wait[1],  a predecessor may pass only the local lock (having already released all its ancestral locks)  to \thr{} by  updating \texttt{q.status} to \parent[1]  (\edge{W_1}{P_1}).
IN \parent[1], when \thr{} notices that it owns the lock at that level, it will update the status to \cohort[1] to indicate the beginning of a new cohort (\edge{P_1}{C_1}).

In \texttt{V/P}$_1$,  \thr{} may attempt to re-acquire the lock, which transitions it to \wait[3] (\edge{V/P_1}{W_2}). In this state, \thr{} will have to wait till the node is recycled.
If \thr{} times out while waiting for the status to become \recycled[] in \wait[3], it will update the status to \parent[2] and leave (\edge{W_3}{P_2}).
In \wait[3], if the predecessor \pred{} trying to pass the lock becomes impatient because a successor \succe{}  has not updated \texttt{q.next}, \pred{} leaves \texttt{q} by updating its status to \parent[3]  (\edge{W_3}{P_3}).
If \thr{} times out while waiting for the status to become \recycled[] in \parent[3], it will update the status to \parent[2] and leave (\edge{P_3}{P_2}).
In \parent[3], if \succe{} advertises itself and recycles \texttt{q.status}, the NFA non-deterministically transitions to either \recycled[3] (\edge{P_3}{R_3}, if it finds no predecessor by the time \thr{} re-enqueues the node) or to \recycled[2] (\edge{P_3}{R_2}, if a predecessor is present by the time \thr{} re-enqueued the node).

In \wait[3], if the predecessor \pred{} manages to successfully release the lock to some other thread or relinquish the lock, \pred{} it will eventually transition \texttt{q.status} to \recycled[3] (\edge{W_3}{R_3}, if \thr{} finds no predecessor by the time it re-enqueues the node) or to \recycled[2] (\edge{W_3}{R_2}, if a predecessor is present by the time \thr{} re-enqueues the node).

In \recycled[2], \thr{} will reenqueue the \qnode{} and it may inherit the global lock (transition to \texttt{V}$_1$, \edge{R_2}{V_1}) or inherit only  lock prefix (transition to \parent[1] , \edge{R_2}{P_1}) from one of its predecessors.
In \recycled[2], \thr{} may timeout and abandon while waiting for the lock (\edge{R_2}{A_1}).

\section{NFA for the next field of a QNode}
\label{appx:NFANext}
We now describe the state diagram for the \texttt{next} field.
The \texttt{next} field starts with a \NULL{} value in state $0_1$.
At the beginning of an acquisition, thread \thr{} transitions to $0_2$, where the value of the \texttt{next} field remains unchanged from before (\edge{0_1}{0_2}).
If  \thr{} finishes relinquishing the lock, the state reverts to $0_1$ (\edge{0_2}{0_1}).
This transition can happen either by \thr{} itself (black solid edge) or after \thr{} has abandoned, which case a predecessor may act on \thr{}'s behalf (blue colored dotted edge).

If a successor enqueues and advertises itself with a legal \qnode{} pointer value \texttt{S}, NFA transitions to $S_1$  (\edge{0_2}{S_1}).
\thr{} may successfully acquire the lock and release, which leaves it in $S_1$.
\thr{} may timeout  and abandon, which leaves it in $S_1$ and subsequent attempts to acquire by \thr{} will leave it in $S_1$ until a predecessor marks the \qnode{} for recycling at which point \thr{} resets the next pointer to \NULL{} just before enqueuing (\edge{S_1}{0_2}).
In $S_1$, if \thr{} times out, a predecessor, may reuse the \texttt{next} field to remember the predecessor on its forward journey to find a waiting successor (\edge{S_1}{P_1}).
In $S_1$, if \thr{} attempts to re-acquire, it will wait and possibly timeout (\edge{S_1}{S_1}). 
In $P_1$, once a predecessor has recycled the \qnode{}, \thr{} will reset the next pointer to \NULL{} and re-enqueue (\edge{P_1}{0_2}). 
In $P_1$, if \thr{} attempts to re-acquire, it will wait and possibly timeout (\edge{P_1}{P_1}). 
In $0_2$, if \thr{} timeouts during release waiting for the successor to update the \texttt{next} pointer, \thr{} writes $M_1$ (\edge{0_2}{M_1}). 
If \thr{} times out during acquire in $0_2$, a predecessor may trigger the \edge{0_2}{M_1} transition. 
In $M_1$, if \thr{} attempts to re-acquire, it will wait and possibly timeout (\edge{M_1}{M_1}) until the node is recycled by the successor (\edge{M_1}{S_1}).

\bibliographystyle{abbrv}
\balance
\bibliography{bibliography}  

%
%
\end{document}